\documentclass[aps,pra,reprint,superscriptaddress,longbibliography]{revtex4-2}

\usepackage{amsmath,amssymb,amsthm,mathtools}
\usepackage{bm}
\usepackage{braket}
\usepackage{microtype}
\usepackage[hidelinks]{hyperref}
\usepackage{centernot}
\usepackage{empheq}
\usepackage{xcolor}

\newcommand{\id}{\mathbb I}
\newcommand{\HSS}{\mathcal H}
\newcommand{\HS}{\mathrm{HS}}
\newcommand{\Tr}{\operatorname{Tr}}

\newtheorem{theorem}{Theorem}[section]
\newtheorem{proposition}[theorem]{Proposition}
\newtheorem{corollary}[theorem]{Corollary}
\newtheorem{lemma}[theorem]{Lemma}
\theoremstyle{remark}
\newtheorem{remark}[theorem]{Remark}
\newtheorem{definition}[theorem]{Definition}

\newcommand*\widefbox[1]{%
	\fbox{\hspace{0.5em}#1\hspace{0.5em}}%
}

\begin{document}
	
	\title{
		Contractivity of the Hilbert--Schmidt Speed and
		Unitality--Divisibility-Based Witnesses
		in Finite-Dimensional Quantum Dynamics
	}

	\author{Hossein Rangani Jahromi}
	\affiliation{Department of Physics, Jahrom University, Jahrom, Iran}
	
\begin{abstract}
	We investigate the Hilbert--Schmidt speed as a witness of non-Markovianity in finite-dimensional quantum dynamics. For a differentiable one-parameter family of quantum states, the Hilbert--Schmidt speed is defined, up to a conventional factor, as the Hilbert--Schmidt norm of the corresponding Hermitian traceless tangent operator. We show that, for unital dynamics in arbitrary finite dimension, P-divisibility implies the monotonic decrease of the Hilbert--Schmidt speed. Hence, any increase in this quantity signals the breakdown of P-divisibility, and therefore also excludes CP-divisibility. For qubit systems, we establish a stronger result: every P-divisible evolution decreases the Hilbert--Schmidt speed, independently of unitality. Thus, in dimension two, HSS growth is a valid witness of non-Markovian behaviour for arbitrary positive divisible dynamics. This conclusion is dimension dependent. In dimensions $d\geq3$, non-unitality can generate HSS growth even under CP-divisible dynamics. We demonstrate this by constructing an explicit non-unital CP-divisible qutrit semigroup for which the Hilbert--Schmidt speed strictly increases. Finally, for unital GKSL dynamics, we derive a generator-level dissipation identity explaining the monotonic decay of the Hilbert--Schmidt speed. These results specify the regimes in which HSS growth can be interpreted as evidence of the failure of divisibility and clarify its limitation for higher-dimensional non-unital evolutions.
\end{abstract}

	\maketitle
	
	\section{Introduction}
	\label{sec:introduction}
	
	Contractive quantities play a central role in the analysis of open
	quantum-system dynamics \cite{BreuerPetruccione2002,NielsenChuang2010}.  If a distinguishability measure decreases under every admissible quantum channel, its temporary increase along a
	dynamical evolution can reveal a failure of divisibility and may
	therefore be used as a witness of memory effects and non-Markovianity \cite{RHP,BreuerReview,RivasReview}.
	
	The Hilbert--Schmidt norm and distance are particularly simple to
	evaluate, but, unlike the trace norm and distance, they are not
	contractive in general for open quantum-system dynamics
	\cite{WangSchirmer2009}. This fact requires care when the
	Hilbert--Schmidt speed is used to diagnose non-Markovian behavior. One
	must distinguish at least three logically different statements:
	\begin{enumerate}
		\item contraction of a single dynamical map between time zero and
		time $t$;
		\item contraction of every intermediate propagator between
		arbitrary times $s$ and $t$;
		\item contraction on all Hermitian operators, as opposed to
		contraction only on traceless Hermitian state tangents.
	\end{enumerate}
	
	Hilbert--Schmidt speed (HSS) has recently emerged as a practical
	witness of non-Markovianity.  In particular, the HSS-based witness was
	introduced in Ref.~\cite{RanganiJahromi2020}, where it was shown,
	through several paradigmatic open-system models, to detect memory
	effects consistently with the trace-distance-based BLP criterion
	\cite{BreuerLainePiilo2009}, while
	retaining the computational advantage that no diagonalization of the
	evolved density matrix is required.  This perspective was subsequently
	extended to high-dimensional and multipartite finite-dimensional
	systems in Ref.~\cite{Mahdavipour2022}, where further evidence was
	found for the effectiveness of HSS in identifying memory effects.
	Beyond conventional open-system dynamics, HSS has also been employed as
	an efficiently computable probe of exceptional points in non-Hermitian
	$PT$- and anti-$PT$-symmetric systems \cite{JahromiLoFranco2022}.
	\par
	These developments naturally motivate a general analytical
	characterization of the dynamical conditions under which HSS must be
	monotone.
	
	The purpose of this work is to clarify these distinctions and establish
	the exact role of unitality and dimension.  Our main results are:
	\begin{enumerate}
		\item A positive trace-preserving map is Hilbert--Schmidt
		contractive on all Hermitian operators if and only if it is unital.
		\item Unital dynamical maps imply contraction relative to the
		initial time.
		\item Unitality together with P-divisibility implies monotonicity
		between every pair of times.
		\item For qubits, P-divisibility alone implies monotonicity of the
		Hilbert--Schmidt speed.
		\item In dimensions $d\geq3$, non-unital CP-divisible dynamics can
		increase the Hilbert--Schmidt speed.
	\end{enumerate}
	
	Throughout the paper, ``P-divisible'' and ``CP-divisible''  refer to
	divisibility by positive trace-preserving and completely positive
	trace-preserving propagators \cite{BreuerPetruccione2002}, respectively.  We avoid using	``Markovian'' without qualification because several inequivalent
	notions of quantum Markovianity occur in the literature.
	
	The paper is organized as follows. In Sec.~\ref{sec:hss-tangent}, we
	define the Hilbert--Schmidt speed and relate its contractivity to
	Hilbert--Schmidt-norm contractivity of physical tangent operators. In
	Sec.~\ref{sec:unitality-contractivity}, we prove the equivalence between
	unitality and Hilbert--Schmidt contractivity on the full Hermitian
	operator space. In Sec.~\ref{sec:dynamical-maps}, we apply this result
	to dynamical maps and distinguish endpoint contraction from full
	monotonicity under P-divisibility. In Sec.~\ref{sec:qubit}, we prove
	the special qubit result showing that unitality is not required for
	Hilbert--Schmidt-speed monotonicity under P-divisibility. In
	Sec.~\ref{sec:gksl-generator}, we derive the generator-level
	dissipation identity for unital GKSL dynamics. In
	Sec.~\ref{sec:qutrit}, we present a non-unital CP-divisible qutrit
	semigroup for which the Hilbert--Schmidt speed increases. In
	Sec.~\ref{sec:consequences}, we discuss the consequences for using HSS
	growth as a divisibility witness, and in
	Sec.~\ref{sec:logical-hierarchy} we summarize the logical hierarchy of
	the results. We close with the conclusions in
	Sec.~\ref{sec:conclusions}.
	
	\section{Hilbert--Schmidt speed and physical tangent operators}
	\label{sec:hss-tangent}
	
	Let $\mathcal H_d$ be a $d$-dimensional Hilbert space and let
	$\mathcal B(\mathcal H_d)$ denote its operator algebra.  The
	Hilbert--Schmidt inner product and norm are
	\begin{equation}
		\langle A,B\rangle_{\HS}
		=
		\Tr(A^\dagger B),
		\qquad
		\|A\|_{\HS}
		=
		\sqrt{\Tr(A^\dagger A)}.
	\end{equation}
	
	Consider a differentiable family of density operators
	$\rho_\varphi$, parametrized by a real parameter $\varphi$.  Its tangent
	operator is
	\begin{equation}
		X_\varphi
		:=
		\partial_\varphi\rho_\varphi.
	\end{equation}
	Because $\rho_\varphi$ is Hermitian and has unit trace,
	\begin{equation}
		X_\varphi^\dagger=X_\varphi,
		\qquad
		\Tr X_\varphi=0.
	\end{equation}
	
	We use the conventional definition
	\begin{equation}
		\HSS(\rho_\varphi)
		:=
		\sqrt{
			\frac{1}{2}
			\Tr\left[
			\left(\partial_\varphi\rho_\varphi\right)^2
			\right]
		}
		=
		\frac{1}{\sqrt 2}\|X_\varphi\|_{\HS}.
		\label{eq:HSS-definition}
	\end{equation}
	The factor $1/\sqrt2$ is conventional and does not affect any
	contractivity statement.
	
	Every traceless Hermitian operator can be realized as a physical local
	state tangent.  Indeed, if $X=X^\dagger$ and $\Tr X=0$, then
	\begin{equation}
		\rho_\varphi
		=
		\frac{\id}{d}+\varphi X
		\label{eq:local-state-family}
	\end{equation}
	is a density operator for all sufficiently small $|\varphi|$.
	Therefore, contractivity on all physical local tangents is equivalent
	to contractivity on the real vector space of traceless Hermitian
	operators.
	
	Let $\Lambda$ be a linear map independent of $\varphi$.  Then
	\begin{equation}
		\partial_\varphi\Lambda(\rho_\varphi)
		=
		\Lambda(\partial_\varphi\rho_\varphi)
		=
		\Lambda(X_\varphi).
		\label{eq:parameter-independent}
	\end{equation}
	Thus, Hilbert--Schmidt-speed contractivity reduces to
	Hilbert--Schmidt-norm contractivity on traceless Hermitian operators.
	
	\begin{remark}
		The assumption that the channel is independent of the encoded parameter
		is essential.  If $\Lambda=\Lambda_\varphi$, then
		\begin{equation}
			\partial_\varphi\Lambda_\varphi(\rho_\varphi)
			=
			(\partial_\varphi\Lambda_\varphi)(\rho_\varphi)
			+
			\Lambda_\varphi(\partial_\varphi\rho_\varphi),
		\end{equation}
		and the additional term prevents a direct contractivity argument.
	\end{remark}
	
	\section{Unitality and Hilbert--Schmidt contractivity}
	\label{sec:unitality-contractivity}
	
	We begin with a characterization on the full real vector space of
	Hermitian operators.
	
	\begin{theorem}
		\label{thm:unital-iff-full-Hermitian}
		Let
		\[
		\Phi:\mathcal B(\mathcal H_d)\longrightarrow
		\mathcal B(\mathcal H_d)
		\]
		be a positive trace-preserving linear map.  Then the following
		statements are equivalent:
		\begin{enumerate}
			\item $\Phi$ is unital:
			\[
			\Phi(\id)=\id.
			\]
			\item $\Phi$ is Hilbert--Schmidt contractive on every Hermitian
			operator:
			\[
			\|\Phi(X)\|_{\HS}\leq\|X\|_{\HS}
			\quad
			\text{for every }X=X^\dagger.
			\]
		\end{enumerate}
	\end{theorem}
	
	\begin{proof}
		Suppose first that $\Phi$ is positive and unital.  Kadison's inequality \cite{Kadison1952} (see Appendix~\ref{app:KS})
		for a unital positive map gives, for every Hermitian $X$, 
		\begin{equation}
			\Phi(X^2)\geq\Phi(X)^2.
			\label{eq:Kadison}
		\end{equation}
		Taking the trace and using trace preservation,
		\begin{align}
			\|\Phi(X)\|_{\HS}^2
			&=
			\Tr[\Phi(X)^2]
			\\
			&\leq
			\Tr[\Phi(X^2)]
			\\
			&=
			\Tr(X^2)
			=
			\|X\|_{\HS}^2.
		\end{align}
		This proves Hilbert--Schmidt contractivity.
		
		Conversely, suppose that $\Phi$ is Hilbert--Schmidt contractive on every
		Hermitian operator.  Apply the assumption to $X=\id$ and define
		\[
		A:=\Phi(\id).
		\]
		Since $\Phi$ is positive and $\id\ge 0$, we have $A=\Phi(\id)\ge 0$.
		Trace preservation gives
		\begin{equation}
			\Tr A=\Tr\id=d.
		\end{equation}
		Contractivity gives
		\begin{equation}
			\Tr(A^2)=\|A\|_{\HS}^2
			\leq
			\|\id\|_{\HS}^2=d.
			\label{eq:A-upper}
		\end{equation}
		On the other hand, the Hilbert--Schmidt Cauchy--Schwarz inequality
		implies
		\begin{equation}
			(\Tr A)^2
			=
			|\Tr(\id A)|^2
			\leq
			\Tr(\id^2)\Tr(A^2)
			=
			d\,\Tr(A^2).
		\end{equation}
		Since $\Tr A=d$, it follows that
		\begin{equation}
			\Tr(A^2)\geq d.
			\label{eq:A-lower}
		\end{equation}
		Equations \eqref{eq:A-upper} and \eqref{eq:A-lower} imply
		$\Tr(A^2)=d$, so equality holds in Cauchy--Schwarz.  Hence $A$ is
		proportional to the identity.  Its trace fixes the proportionality
		constant:
		\begin{equation}
			A=\id.
		\end{equation}
		Therefore, $\Phi$ is unital.
	\end{proof}
	
	\begin{remark}
		The distinction between all Hermitian operators and physical state
		tangents is essential.  Physical state tangents form the smaller
		subspace of traceless Hermitian operators.  Therefore,
		Theorem~\ref{thm:unital-iff-full-Hermitian} implies that every unital
		positive trace-preserving map contracts physical tangents, but the
		converse need not hold: contractivity only on physical tangents does
		not imply unitality.  Indeed, the proof of necessity in the theorem
		tests contractivity on $\id$, which is not traceless and hence is not
		a physical state tangent.  For qubits, every positive trace-preserving
		map contracts all traceless Hermitian operators, including maps that
		are non-unital; see Sec.~\ref{sec:qubit}.
	\end{remark}
	
	\begin{corollary}
		\label{cor:unital-HSS}
		Let $\Phi$ be a positive, trace-preserving, unital map independent of
		$\varphi$.  Then
		\begin{equation}
			\HSS\bigl(\Phi(\rho_\varphi)\bigr)
			\leq
			\HSS(\rho_\varphi)
		\end{equation}
		for every differentiable family of states $\rho_\varphi$.
	\end{corollary}
	
	\begin{proof}
		The tangent $X_\varphi=\partial_\varphi\rho_\varphi$ is Hermitian.
		Equation \eqref{eq:parameter-independent} and
		Theorem~\ref{thm:unital-iff-full-Hermitian} therefore give
		\begin{align}
			\HSS\bigl(\Phi(\rho_\varphi)\bigr)
			&=
			\frac{1}{\sqrt2}\|\Phi(X_\varphi)\|_{\HS}
			\\
			&\leq
			\frac{1}{\sqrt2}\|X_\varphi\|_{\HS}
			=
			\HSS(\rho_\varphi).
		\end{align}
	\end{proof}
	
	\section{Dynamical maps, divisibility, and propagator unitality}
	\label{sec:dynamical-maps}
	
	Consider a family of positive trace-preserving dynamical maps
	$\{\Lambda_t\}_{t\geq0}$ acting on a fixed finite-dimensional operator
	algebra, with
	\begin{equation}
		\Lambda_0=\operatorname{id}.
	\end{equation}
	The dynamics is assumed to be independent of the parameter $\varphi$.
	
	\begin{definition}
		The evolution is P-divisible if, for every $t\geq s\geq0$, there exists
		a positive trace-preserving map $V_{t,s}$ such that \cite{NielsenChuang2010}
		\begin{equation}
			\Lambda_t=V_{t,s}\circ\Lambda_s.
			\label{eq:P-divisibility}
		\end{equation}
		It is CP-divisible if the maps $V_{t,s}$ can be chosen completely
		positive and trace preserving.
	\end{definition}
	
	When $\Lambda_s$ is non-invertible, a factor map satisfying
	\eqref{eq:P-divisibility} need not be unique.  None of the following
	results requires invertibility or uniqueness.
	
	\begin{proposition}
		\label{prop:unital-equivalence}
		For a P-divisible evolution with $\Lambda_0=\operatorname{id}$, the
		following statements are equivalent:
		\begin{enumerate}
			\item Every dynamical map is unital:
			\begin{equation}
				\Lambda_t(\id)=\id
				\quad\text{for every }t\geq0.
				\label{eq:endpoint-unitality}
			\end{equation}
			
			\item For every $t\geq s\geq0$, there exists a unital positive
			trace-preserving map $V_{t,s}$ satisfying
			\begin{equation}
				\Lambda_t=V_{t,s}\circ\Lambda_s.
			\end{equation}
			
			\item For every $t\geq s\geq0$, every positive trace-preserving map
			$V_{t,s}$ satisfying
			\begin{equation}
				\Lambda_t=V_{t,s}\circ\Lambda_s
			\end{equation}
			is unital.
		\end{enumerate}
	\end{proposition}
	
	\begin{proof}
		Assume that statement 1 holds. Thus, for every \(r\ge 0\),
		\[
		\Lambda_r(\id)=\id .
		\]
		Fix \(t\ge s\ge 0\), and let \(V_{t,s}\) be any admissible factor map satisfying
		\[
		\Lambda_t=V_{t,s}\circ \Lambda_s .
		\]
		Since \(\Lambda_s\) is unital, \(\Lambda_s(\id)=\id\). Hence
		\[
		V_{t,s}(\id)
		=
		V_{t,s}\bigl(\Lambda_s(\id)\bigr)
		=
		\Lambda_t(\id)
		=
		\id .
		\]
		Therefore every admissible propagator \(V_{t,s}\) is unital.
		
		Statement 3 immediately implies statement 2 because P-divisibility
		guarantees the existence of at least one positive trace-preserving
		factor map for each pair of times.
		
		Finally, assume statement 2 and choose $s=0$.  Since
		$\Lambda_0=\operatorname{id}$,
		\begin{equation}
			\Lambda_t
			=
			V_{t,0}\circ\Lambda_0
			=
			V_{t,0}.
		\end{equation}
		The map $V_{t,0}$ is unital, and therefore $\Lambda_t$ is unital.
		This proves statement 1.
	\end{proof}
	
	\begin{remark}
		The proof of Proposition~\ref{prop:unital-equivalence} does not define
		$V_{t,s}$ by $\Lambda_t\Lambda_s^{-1}$ and does not assume that
		$\Lambda_s$ is invertible.  If the dynamical maps are unital, the
		calculation
		\[
		V_{t,s}(\id)
		=
		V_{t,s}(\Lambda_s(\id))
		=
		\Lambda_t(\id)
		=
		\id
		\]
		applies separately to every admissible factor map.  Propagator
		non-uniqueness therefore creates no ambiguity in the conclusion.
	\end{remark}
	
	The next result separates endpoint contraction from full temporal
	monotonicity.
	
	\begin{theorem}
		\label{thm:endpoint-vs-monotonicity}
		Let $\rho_\varphi$ be a differentiable family of states and define
		\[
		\rho_\varphi(t)=\Lambda_t(\rho_\varphi),
		\qquad
		X_\varphi(t)=\partial_\varphi\rho_\varphi(t).
		\]
		Then:
		\begin{enumerate}
			\item If every $\Lambda_t$ is positive, trace preserving, and
			unital, then
			\begin{equation}
				\HSS(t)\leq\HSS(0)
				\quad\text{for every }t\geq0.
				\label{eq:endpoint-contraction}
			\end{equation}
			
			\item If, in addition, the evolution is P-divisible, then
			\begin{equation}
				\HSS(t)\leq\HSS(s)
				\quad\text{for every }t\geq s\geq0.
				\label{eq:full-monotonicity}
			\end{equation}
		\end{enumerate}
	\end{theorem}
	
	\begin{proof}
		Parameter independence gives
		\begin{equation}
			X_\varphi(t)=\Lambda_t(X_\varphi(0)).
		\end{equation}
		If $\Lambda_t$ is unital, Corollary~\ref{cor:unital-HSS} gives
		\[
		\HSS(t)\leq\HSS(0).
		\]
		
		Now suppose that the evolution is P-divisible.  By
		Proposition~\ref{prop:unital-equivalence}, every admissible positive
		trace-preserving propagator is unital.  Moreover,
		\begin{align}
			X_\varphi(t)
			&=
			\partial_\varphi
			\left[
			V_{t,s}(\rho_\varphi(s))
			\right]
			\\
			&=
			V_{t,s}(X_\varphi(s)).
		\end{align}
		Applying Corollary~\ref{cor:unital-HSS} to $V_{t,s}$ gives
		\[
		\HSS(t)
		=
		\frac{1}{\sqrt2}\|V_{t,s}(X_\varphi(s))\|_{\HS}
		\leq
		\frac{1}{\sqrt2}\|X_\varphi(s)\|_{\HS}
		=
		\HSS(s).
		\]
	\end{proof}
	
	\begin{remark}
		Unitality of the endpoint maps alone gives only
		$\HSS(t)\leq\HSS(0)$.  It does not by itself compare
		$\HSS(t)$ with $\HSS(s)$ for arbitrary positive times $t\geq s$.
		P-divisibility supplies the physical propagator needed for that
		comparison.
	\end{remark}
	
	\begin{corollary}
		For a unital evolution, an interval on which
		\[
		\frac{d}{dt}\HSS(t)>0
		\]
		rules out P-divisibility on that interval.
	\end{corollary}
	
	\section{The special case of qubits}
	\label{sec:qubit}
	
	For qubits, traceless Hermitian operators have a special relation
	between their trace and Hilbert--Schmidt norms.  This leads to a
	stronger result that does not require unitality.
	
	We first establish the trace-norm property needed in the proof.
	
	\begin{lemma}
		\label{lem:trace-contractivity}
		Let $\Phi$ be a positive trace-preserving map.  Then
		\begin{equation}
			\|\Phi(X)\|_1\leq\|X\|_1
		\end{equation}
		for every Hermitian operator $X$.
	\end{lemma}
	
	\begin{proof}
		Let
		\begin{equation}
			X=X_+-X_-
		\end{equation}
		be the Jordan decomposition of $X$, where
		\begin{equation}
			X_\pm\geq0,
			\qquad
			X_+X_-=0.
		\end{equation}
		Positivity of $\Phi$ implies
		\[
		\Phi(X_\pm)\geq0.
		\]
		Using the triangle inequality, positivity, and trace preservation,
		\begin{align}
			\|\Phi(X)\|_1
			&=
			\|\Phi(X_+)-\Phi(X_-)\|_1
			\\
			&\leq
			\|\Phi(X_+)\|_1+\|\Phi(X_-)\|_1
			\\
			&=
			\Tr\Phi(X_+)+\Tr\Phi(X_-)
			\\
			&=
			\Tr X_+ + \Tr X_-
			\\
			&=
			\|X\|_1.
		\end{align}
	\end{proof}
	
	\begin{remark}
		Lemma~\ref{lem:trace-contractivity} is stated only for Hermitian
		operators.  No trace-norm contractivity claim for arbitrary
		non-Hermitian inputs is needed.
	\end{remark}
	
	\begin{theorem}
		\label{thm:qubit-contraction}
		Let
		\[
		\Phi:\mathcal B(\mathbb C^2)\longrightarrow
		\mathcal B(\mathbb C^2)
		\]
		be positive and trace preserving.  Then
		\begin{equation}
			\|\Phi(X)\|_{\HS}\leq\|X\|_{\HS}
		\end{equation}
		for every traceless Hermitian operator $X$.
	\end{theorem}
	
	\begin{proof}
		First, positivity implies that $\Phi$ preserves Hermiticity.  Indeed,
		every Hermitian operator is the difference of two positive operators,
		and the images of those positive operators are Hermitian.
		
		Trace preservation gives
		\begin{equation}
			\Tr\Phi(X)=\Tr X=0.
		\end{equation}
		Thus both $X$ and $\Phi(X)$ are traceless Hermitian operators on a
		two-dimensional Hilbert space.
		
		The eigenvalues of a traceless Hermitian qubit operator have the form
		$\{\lambda,-\lambda\}$.  Therefore,
		\begin{equation}
			\|X\|_1=2|\lambda|,
			\qquad
			\|X\|_{\HS}=\sqrt2|\lambda|,
		\end{equation}
		and hence
		\begin{equation}
			\|X\|_1=\sqrt2\,\|X\|_{\HS}.
			\label{eq:qubit-norm-relation}
		\end{equation}
		The same identity holds for $\Phi(X)$.  Lemma
		\ref{lem:trace-contractivity} consequently gives
		\begin{align}
			\sqrt2\,\|\Phi(X)\|_{\HS}
			&=
			\|\Phi(X)\|_1
			\\
			&\leq
			\|X\|_1
			=
			\sqrt2\,\|X\|_{\HS}.
		\end{align}
		Dividing by $\sqrt2$ proves the result.
	\end{proof}
	
	\begin{corollary}
		\label{cor:qubit-P-divisible}
		For every parameter-independent P-divisible qubit evolution,
		\begin{equation}
			\HSS(t)\leq\HSS(s)
			\quad
			\text{for all }t\geq s\geq0,
		\end{equation}
		regardless of whether the dynamics is unital.
	\end{corollary}
	
	\begin{proof}
		For every physical state family, $X_\varphi(s)$ is traceless and
		Hermitian.  P-divisibility gives
		\[
		X_\varphi(t)=V_{t,s}(X_\varphi(s)),
		\]
		where $V_{t,s}$ is positive and trace preserving.  Theorem
		\ref{thm:qubit-contraction} then yields
		\[
		\|X_\varphi(t)\|_{\HS}
		\leq
		\|X_\varphi(s)\|_{\HS}.
		\]
		Equation \eqref{eq:HSS-definition} gives the claimed inequality.
	\end{proof}
	
	\begin{remark}
		For qubits, an increase of the Hilbert--Schmidt speed rules out
		P-divisibility without any unitality assumption.  This conclusion is
		stronger than the corresponding statement in general dimensions.
		\par This qubit exceptionalism is consistent with earlier observations that,
		for two-level systems, the Hilbert--Schmidt distance is always
		monotonically decreasing, unlike the generic situation in dimensions
		greater than two \cite{WangSchirmer2009}.
	\end{remark}
	
	\section{Generator-level contraction for unital GKSL dynamics}
	\label{sec:gksl-generator}
	
	We next derive a differential version of the HSS contraction result for
	time-local GKSL dynamics. In this section we assume that the estimated
	parameter $\varphi$ is encoded only in the initial state, so that the
	generator $\mathcal L_t$ is independent of $\varphi$. If
	\[
	\rho_\varphi(t)
	\]
	satisfies
	\[
	\frac{d}{dt}\rho_\varphi(t)=\mathcal L_t(\rho_\varphi(t)),
	\]
	then the tangent operator
	\[
	X(t):=\partial_\varphi \rho_\varphi(t)
	\]
	obeys the homogeneous equation
	\begin{equation}
		\dot X(t)=\mathcal L_t(X(t)).
		\label{eq:tangent-homogeneous}
	\end{equation}
	We assume throughout that $X(t)$ is Hermitian.
	
	Let the time-local generator have the
	Gorini--Kossakowski--Sudarshan--Lindblad (GKSL) form
	\cite{GKSL2,GKSL1},
	\begin{align}
		\mathcal L_t(X)
		&=
		-i[H_t,X]
		\nonumber\\
		&\quad+
		\sum_\alpha\gamma_\alpha(t)
		\left(
		L_{\alpha,t}XL_{\alpha,t}^\dagger
		-
		\frac12\{L_{\alpha,t}^\dagger L_{\alpha,t},X\}
		\right),
		\label{eq:GKSL}
	\end{align}
	with
	\begin{equation}
		\gamma_\alpha(t)\ge 0.
	\end{equation}
	Nonnegative rates correspond to CP-divisible time-local dynamics.
	
	Define
	\begin{equation}
		\mathcal D_L(X)
		:=
		LXL^\dagger-\frac12\{L^\dagger L,X\}.
	\end{equation}
	
	\begin{lemma}
		\label{lem:dissipator-identity}
		For every Hermitian $X$ and arbitrary operator $L$,
		\begin{equation}
			2\Tr[X\mathcal D_L(X)]
			=
			-\|[L,X]\|_{\HS}^2
			+
			\Tr\!\left[
			X^2(LL^\dagger-L^\dagger L)
			\right].
			\label{eq:general-dissipator-identity}
		\end{equation}
	\end{lemma}
	
	\begin{proof}
		We expand both sides explicitly. Since $X=X^\dagger$,
		\begin{align}
			\|[L,X]\|_{\HS}^2
			&=
			\Tr\!\left[(LX-XL)^\dagger(LX-XL)\right]
			\nonumber\\
			&=
			\Tr\!\left[
			(XL^\dagger-L^\dagger X)(LX-XL)
			\right]
			\nonumber\\
			&=
			\Tr(XL^\dagger LX)
			-\Tr(XL^\dagger XL)
			\nonumber\\
			&\quad
			-\Tr(L^\dagger XLX)
			+\Tr(L^\dagger X^2L).
			\label{eq:commutator-expanded-raw}
		\end{align}
		By cyclicity of the trace,
		\begin{align}
			\Tr(XL^\dagger LX)
			&=
			\Tr(L^\dagger L X^2),
			\\
			\Tr(L^\dagger X^2L)
			&=
			\Tr(LL^\dagger X^2),
			\\
			\Tr(XL^\dagger XL)
			&=
			\Tr(LXL^\dagger X),
			\\
			\Tr(L^\dagger XLX)
			&=
			\Tr(XLXL^\dagger).
		\end{align}
		The two cross terms are equal by cyclicity:
		\[
		\Tr(LXL^\dagger X)=\Tr(XLXL^\dagger).
		\]
		Therefore
		\begin{equation}
			\|[L,X]\|_{\HS}^2
			=
			\Tr(L^\dagger L X^2)
			+
			\Tr(LL^\dagger X^2)
			-
			2\Tr(XLXL^\dagger).
			\label{eq:commutator-expansion}
		\end{equation}
		
		On the other hand, from the definition of $\mathcal D_L$,
		\begin{align}
			2\Tr[X\mathcal D_L(X)]
			&=
			2\Tr(XLXL^\dagger)
			-
			\Tr\!\left(XL^\dagger LX\right)
			-
			\Tr\!\left(X^2L^\dagger L\right)
			\nonumber\\
			&=
			2\Tr(XLXL^\dagger)
			-
			2\Tr(L^\dagger L X^2),
			\label{eq:dissipator-expansion}
		\end{align}
		where the second equality again uses cyclicity of the trace.
		
		Using \eqref{eq:commutator-expansion}, we obtain
		\begin{align}
			2\Tr(XLXL^\dagger)
			&=
			-\|[L,X]\|_{\HS}^2
			+
			\Tr(L^\dagger L X^2)
			+
			\Tr(LL^\dagger X^2).
		\end{align}
		Substituting this identity into \eqref{eq:dissipator-expansion} gives
		\begin{align}
			2\Tr[X\mathcal D_L(X)]
			&=
			-\|[L,X]\|_{\HS}^2
			+
			\Tr(LL^\dagger X^2)
			-
			\Tr(L^\dagger L X^2)
			\nonumber\\
			&=
			-\|[L,X]\|_{\HS}^2
			+
			\Tr\!\left[
			X^2(LL^\dagger-L^\dagger L)
			\right],
		\end{align}
		where the last step uses cyclicity of the trace. This proves
		\eqref{eq:general-dissipator-identity}.
	\end{proof}
	
	The generator is unital when
	\begin{equation}
		\mathcal L_t(\id)=0.
	\end{equation}
	For the representation \eqref{eq:GKSL}, this is equivalent to
	\begin{equation}
		\sum_\alpha
		\gamma_\alpha(t)
		\left(
		L_{\alpha,t}L_{\alpha,t}^\dagger
		-
		L_{\alpha,t}^\dagger L_{\alpha,t}
		\right)
		=
		0.
		\label{eq:generator-unitality}
	\end{equation}
	
	\begin{theorem}
		\label{thm:generator-contraction}
		Assume that the parameter $\varphi$ is encoded only in the initial state,
		so that the tangent operator $X(t)=\partial_\varphi\rho_\varphi(t)$
		satisfies \eqref{eq:tangent-homogeneous}. Let the evolution admit the
		time-local GKSL form \eqref{eq:GKSL} with $\gamma_\alpha(t)\ge 0$ and
		$\mathcal L_t(\id)=0$. Then
		\begin{empheq}[box=\widefbox]{equation}
			\frac{d}{dt}\HSS^2(t)
			=
			-\frac12
			\sum_\alpha
			\gamma_\alpha(t)
			\|[L_{\alpha,t},X(t)]\|_{\HS}^2
			\le 0.
			\label{eq:HSS-generator}
		\end{empheq}
		In particular, $\HSS^2(t)$ is nonincreasing in time. Since
		$\HSS(t)\ge 0$ and the square-root function is increasing on
		$[0,\infty)$, it follows that $\HSS(t)$ itself is nonincreasing.
		Equivalently, for $t\ge s$,
		\[
		\HSS^2(t)\le \HSS^2(s)
		\quad\Longrightarrow\quad
		\HSS(t)\le \HSS(s).
		\]
		Thus every unital CP-divisible GKSL evolution is Hilbert--Schmidt-speed
		contractive along the tangent flow:
		\[
		t\ge s
		\quad\Longrightarrow\quad
		\HSS(t)\le \HSS(s).
		\]
	\end{theorem}
	
	\begin{proof}
		By definition,
		\begin{equation}
			\HSS^2(t)
			=
			\frac12\Tr[X(t)^2],
			\label{eq:HSS-square-definition}
		\end{equation}
		because $X(t)$ is Hermitian. Differentiating
		\eqref{eq:HSS-square-definition} gives
		\begin{align}
			\frac{d}{dt}\HSS^2(t)
			&=
			\frac12\frac{d}{dt}\Tr[X(t)^2]
			\nonumber\\
			&=
			\frac12\Tr[\dot X(t)X(t)+X(t)\dot X(t)]
			\nonumber\\
			&=
			\Tr[X(t)\dot X(t)].
			\label{eq:HSS-square-derivative}
		\end{align}
		In the last equality we used cyclicity of the trace.
		
		Since the parameter $\varphi$ is encoded only in the initial state, the
		generator is independent of $\varphi$. Therefore
		\[
		\dot X(t)=\mathcal L_t(X(t)),
		\]
		and hence
		\begin{equation}
			\frac{d}{dt}\HSS^2(t)
			=
			\Tr[X(t)\mathcal L_t(X(t))].
			\label{eq:HSS-generator-start}
		\end{equation}
		
		We now insert the GKSL form \eqref{eq:GKSL}. The Hamiltonian part gives
		\begin{align}
			\Tr\!\left[
			X(t)\bigl(-i[H_t,X(t)]\bigr)
			\right]
			&=
			-i\Tr\!\left[
			X(t)H_tX(t)-X(t)^2H_t
			\right]
			\nonumber\\
			&=
			-i\Tr\!\left[
			H_tX(t)^2-X(t)^2H_t
			\right]
			\nonumber\\
			&=0,
			\label{eq:hamiltonian-vanishes}
		\end{align}
		where cyclicity of the trace was used in the second line.
		
		For each dissipative term, Lemma~\ref{lem:dissipator-identity} gives
		\begin{align}
			\Tr\!\left[
			X(t)\mathcal D_{L_{\alpha,t}}(X(t))
			\right]
			&=
			-\frac12
			\|[L_{\alpha,t},X(t)]\|_{\HS}^2
			\nonumber\\
			&\quad+
			\frac12
			\Tr\!\left[
			X(t)^2
			\left(
			L_{\alpha,t}L_{\alpha,t}^\dagger
			-
			L_{\alpha,t}^\dagger L_{\alpha,t}
			\right)
			\right].
			\label{eq:single-dissipator-contribution}
		\end{align}
		Summing over $\alpha$ with weights $\gamma_\alpha(t)$, and using
		\eqref{eq:HSS-generator-start}, we obtain
		\begin{align}
			\frac{d}{dt}\HSS^2(t)
			&=
			-\frac12
			\sum_\alpha
			\gamma_\alpha(t)
			\|[L_{\alpha,t},X(t)]\|_{\HS}^2
			\nonumber\\
			&\quad+
			\frac12
			\Tr\!\left[
			X(t)^2
			\sum_\alpha
			\gamma_\alpha(t)
			\left(
			L_{\alpha,t}L_{\alpha,t}^\dagger
			-
			L_{\alpha,t}^\dagger L_{\alpha,t}
			\right)
			\right].
			\label{eq:before-unitality}
		\end{align}
		The second term vanishes by the unitality condition
		\eqref{eq:generator-unitality}. Therefore
		\begin{equation}
			\frac{d}{dt}\HSS^2(t)
			=
			-\frac12
			\sum_\alpha
			\gamma_\alpha(t)
			\|[L_{\alpha,t},X(t)]\|_{\HS}^2.
		\end{equation}
		Since $\gamma_\alpha(t)\ge 0$ and
		$\|[L_{\alpha,t},X(t)]\|_{\HS}^2\ge 0$, every term in the sum is
		nonnegative before the minus sign. Hence
		\[
		\frac{d}{dt}\HSS^2(t)\le 0.
		\]
		This proves the boxed formula \eqref{eq:HSS-generator}.
		
		It remains only to spell out the contractivity consequence. For any
		$t\ge s$, integrating \eqref{eq:HSS-generator} from $s$ to $t$ yields
		\begin{align}
			\HSS^2(t)-\HSS^2(s)
			&=
			-\frac12
			\int_s^t
			\sum_\alpha
			\gamma_\alpha(\tau)
			\|[L_{\alpha,\tau},X(\tau)]\|_{\HS}^2
			d\tau
			\le 0.
		\end{align}
		Thus $\HSS^2(t)\le \HSS^2(s)$. Since $\HSS(t)\ge 0$ for all $t$, taking
		the square root preserves the inequality:
		\[
		\HSS(t)\le \HSS(s).
		\]
		This is precisely Hilbert--Schmidt-speed contractivity along the tangent
		flow.
	\end{proof}
	
	\begin{remark}
		Theorem~\ref{thm:generator-contraction} does not require the Lindblad
		operators $L_{\alpha,t}$ to be Hermitian, nor does it require each
		individual dissipator to be unital. The required condition is only the
		unitality of the full generator, $\mathcal L_t(\id)=0$.
	\end{remark}
	
	\begin{remark}
		If the generator depends explicitly on the parameter $\varphi$, then
		\[
		\dot X(t)
		=
		\mathcal L_{t,\varphi}(X(t))
		+
		(\partial_\varphi\mathcal L_{t,\varphi})(\rho_\varphi(t)).
		\]
		In this case an additional source term appears. Consequently, the
		derivative of $\HSS^2(t)$ contains an extra contribution with no
		definite sign, and the monotonicity formula \eqref{eq:HSS-generator}
		need not hold.
	\end{remark}
	
	\section{A non-unital CP-divisible qutrit counterexample}
	\label{sec:qutrit}
	
	The qubit result does not extend to dimensions $d\geq3$.  We now give
	an explicit qutrit semigroup for which the Hilbert--Schmidt speed grows,
	even though the dynamics is CP-divisible.
	
	Let $\{\ket{1},\ket{2},\ket{3}\}$ be an orthonormal qutrit basis and
	consider the GKSL generator
	\begin{align}
		\mathcal L(\rho)
		&=
		\gamma
		\left(
		L\rho L^\dagger
		-
		\frac12\{L^\dagger L,\rho\}
		\right),
		\nonumber\\
		&\qquad
		L=\ket{1}\!\bra{2},
		\qquad
		\gamma>0.
		\label{eq:qutrit-generator}
	\end{align}
	This generator describes irreversible population transfer from level
	$2$ to level $1$.  Because it is a time-independent GKSL generator
	with positive rate, the semigroup
	\begin{equation}
		\Lambda_t=e^{t\mathcal L}
	\end{equation}
	is CPTP and CP-divisible.
	
	The generator is not unital:
	\begin{align}
		\mathcal L(\id)
		&=
		\gamma(LL^\dagger-L^\dagger L)
		\\
		&=
		\gamma
		\left(
		\ket{1}\!\bra{1}
		-
		\ket{2}\!\bra{2}
		\right)
		\neq0.
	\end{align}
	
	Consider the local state family
	\begin{equation}
		\rho_\varphi(0)
		=
		\frac{\id_3}{3}
		+
		\varphi X(0),
		\qquad
		X(0)
		=
		\operatorname{diag}(2,1,-3).
		\label{eq:qutrit-family}
	\end{equation}
	This is a valid density operator for
	\begin{equation}
		-\frac16\leq\varphi\leq\frac19,
	\end{equation}
	and its tangent is the traceless Hermitian operator $X(0)$.
	
	Since the channel is independent of $\varphi$,
	\begin{equation}
		X(t)=\Lambda_t(X(0)).
	\end{equation}
	For a diagonal tangent
	\[
	X(t)=\operatorname{diag}(x_1(t),x_2(t),x_3(t)),
	\]
	the master equation gives
	\begin{equation}
		\dot x_1(t)=\gamma x_2(t),
		\qquad
		\dot x_2(t)=-\gamma x_2(t),
		\qquad
		\dot x_3(t)=0.
	\end{equation}
	With
	\[
	(x_1(0),x_2(0),x_3(0))=(2,1,-3),
	\]
	the solution is
	\begin{equation}
		x_1(t)=3-e^{-\gamma t},
		\qquad
		x_2(t)=e^{-\gamma t},
		\qquad
		x_3(t)=-3.
	\end{equation}
	Therefore,
	\begin{equation}
		X(t)
		=
		\operatorname{diag}
		\left(
		3-e^{-\gamma t},
		e^{-\gamma t},
		-3
		\right).
	\end{equation}
	
	Writing $u=e^{-\gamma t}$, the squared Hilbert--Schmidt speed is
	\begin{align}
		\HSS^2(t)
		&=
		\frac12\Tr[X(t)^2]
		\\
		&=
		\frac12\left[(3-u)^2+u^2+9\right]
		\\
		&=
		9-3u+u^2.
		\label{eq:qutrit-HSS}
	\end{align}
	Its derivative is
	\begin{align}
		\frac{d}{dt}\HSS^2(t)
		&=
		\gamma u(3-2u).
		\label{eq:qutrit-growth}
	\end{align}
	For $t\geq0$, one has $0<u\leq1$, and hence
	\begin{equation}
		\frac{d}{dt}\HSS^2(t)>0.
	\end{equation}
	Thus the Hilbert--Schmidt speed grows strictly at every finite time,
	although the evolution is a CPTP, CP-divisible semigroup.  This is
	consistent with the general observation that noncontractivity of
	the Hilbert--Schmidt norm/distance is typical beyond the qubit case
	\cite{WangSchirmer2009}.
	
	\begin{remark}
		The increase in \eqref{eq:qutrit-growth} is caused by non-unitality, not
		by a failure of CP-divisibility.  Consequently, in dimensions
		$d\geq3$, Hilbert--Schmidt-speed growth does not by itself witness
		non-P-divisibility or non-CP-divisibility.
	\end{remark}
	
	\begin{remark}
		The counterexample embeds directly into every dimension $d>3$.  One
		may retain the jump operator $L=\ket{1}\!\bra{2}$, let the generator act
		trivially on the remaining levels, and choose
		\begin{equation}
			X(0)=\operatorname{diag}(2,1,-3,0,\ldots,0).
		\end{equation}
		The local family
		\[
		\rho_\varphi(0)=\frac{\id_d}{d}+\varphi X(0)
		\]
		is physical for sufficiently small $|\varphi|$, and the same
		Hilbert--Schmidt-speed growth follows.
	\end{remark}
	
	\section{Consequences for divisibility witnesses}
	\label{sec:consequences}
	
	The preceding results lead to a dimension- and symmetry-dependent
	interpretation of the growth of the Hilbert--Schmidt speed.
	
	\subsection{Unital evolutions in arbitrary finite dimension}
	\label{subsec:unital-evolutions}
	
	Let us first consider the case in which all dynamical maps are unital. If
	there exist $t>s\geq0$ and an encoded family of states such that
	\begin{equation}
		\HSS(t)>\HSS(s),
	\end{equation}
	then the evolution cannot be P-divisible between the times $s$ and $t$.
	Indeed, if the evolution were P-divisible, then, together with the
	unitality of the endpoint maps, every admissible positive trace-preserving
	propagator would be unital and therefore contractive in the
	Hilbert--Schmidt norm.
	
	Therefore, for unital finite-dimensional dynamics, an increase in the
	Hilbert--Schmidt speed provides a valid witness of the failure of
	P-divisibility. Since CP-divisibility implies P-divisibility, this growth
	also rules out CP-divisibility.
	
	\subsection{Qubit evolutions}
	\label{subsec:qubit-evolutions}
	
	For qubit evolutions, the unitality condition is not required. Every
	positive trace-preserving propagator contracts traceless Hermitian
	tangent operators. Hence, the condition
	\begin{equation}
		\HSS(t)>\HSS(s)
	\end{equation}
	rules out P-divisibility for arbitrary qubit dynamics.
	
	This result is a consequence of a dimension-specific norm identity and
	therefore should not be extended to higher-dimensional systems.
	
	\subsection{Non-unital dynamics in dimensions \(d\geq3\)}
	\label{subsec:nonunital-highd}
	
	In dimensions $d\geq3$, non-unitality may lead to an increase in the
	Hilbert--Schmidt speed even for CP-divisible semigroups. Accordingly,
	such a growth can have two possible origins:
	\begin{enumerate}
		\item failure of P-divisibility; or
		\item failure of Hilbert--Schmidt contractivity on traceless Hermitian
		tangents caused by non-unital dynamics.
	\end{enumerate}
	Thus, independent information about the unitality of the dynamics is
	required before the growth can be interpreted as a witness of
	divisibility breaking.
	
	\section{Logical hierarchy of the results}
	\label{sec:logical-hierarchy}
	
	For clarity, the main implications obtained in this work can be written
	as follows.
	
	For a single positive trace-preserving map, one has
	\begin{equation}
		\begin{aligned}
			\Phi\text{ unital}
			\quad &\Longleftrightarrow\quad
			\|\Phi(X)\|_{\HS}\leq\|X\|_{\HS}
			\\[-2pt]
			&\hspace{1em}
			\text{for every Hermitian }X,
			\\[3pt]
			\Phi\text{ unital}
			\quad &\Longrightarrow\quad
			\|\Phi(X)\|_{\HS}\leq\|X\|_{\HS}
			\\[-2pt]
			&\hspace{1em}
			\text{for every physical tangent }X.
		\end{aligned}
	\end{equation}
	
	The second implication is not, in general, an equivalence, because
	physical tangents are necessarily traceless.
	
	For a parameter-independent dynamical evolution, unitality of the
	dynamical maps at all times implies
	\begin{equation}
		\Lambda_t\text{ unital for every }t
		\quad\Longrightarrow\quad
		\HSS(t)\leq\HSS(0).
	\end{equation}
	If, in addition, the evolution is P-divisible, then the stronger
	monotonicity relation
	\begin{equation}
		\HSS(t)\leq\HSS(s)
		\qquad
		(t\geq s\geq0)
	\end{equation}
	is obtained.
	
	For qubit systems, one finds that
	\begin{equation}
		\text{P-divisibility}
		\quad\Longrightarrow\quad
		\HSS(t)\leq\HSS(s),
	\end{equation}
	without imposing the unitality condition.
	
	For $d\geq3$, however,
	\begin{equation}
		\text{CP-divisibility}
		\centernot\Longrightarrow
		\HSS(t)\leq\HSS(s)
	\end{equation}
	when non-unital dynamics is allowed.
	
\section{Conclusions}
\label{sec:conclusions}

In this work, we studied the Hilbert--Schmidt speed as a witness of
non-Markovianity in finite-dimensional open quantum systems. We
identified the conditions under which an increase of the HSS provides a
reliable signature of the breakdown of divisibility. In particular, we
showed that the validity of this witness is closely related to the
unitality of the evolution and the dimension of the system.

For unital dynamics, P-divisibility guarantees that the HSS decreases
monotonically in time. Therefore, any increase of the HSS directly
signals the breakdown of P-divisibility, and consequently also excludes
CP-divisibility. Thus, for unital evolutions in arbitrary finite
dimension, the HSS-based witness provides a clear signature of
non-Markovian behavior.

We also showed that the qubit case is more general. For
two-dimensional systems, the HSS decreases under every P-divisible
evolution, without requiring unitality. Hence, for arbitrary qubit
dynamics, an increase of the HSS is a valid witness of the failure of
P-divisibility. This confirms the robustness of the HSS-based witness
for qubit open systems.

In higher dimensions, non-unitality requires additional attention. We
constructed a non-unital CP-divisible qutrit semigroup for which the HSS
increases. This result does not weaken the HSS-based witness, but
clarifies its proper interpretation beyond the qubit case. In
particular, for non-unital systems of dimension three and higher, an
increase of the HSS should be considered together with the unitality
properties of the evolution.

Finally, for unital GKSL dynamics, we derived a dissipation relation
that directly explains the monotonic decrease of the HSS. This result
connects the HSS-based approach with the generator description of open
quantum dynamics and supports its use as a witness of the breakdown of
Markovian evolution.

Overall, our results establish the HSS as a strong and reliable witness
of non-Markovianity for unital dynamics in arbitrary finite dimension
and for general qubit dynamics. At the same time, they determine how
the witness should be interpreted for higher-dimensional non-unital
systems, thereby providing a precise framework for its application.

	\section*{Data availability}
	
	No data were generated or analyzed in this theoretical study.
	
	\section*{Conflict of interest}
	
	The author declares no conflict of interest.

\appendix

\section{Positive maps and Kadison's inequality}
\label{app:KS}

We recall a standard proof of the self-adjoint form of Kadison's
inequality \cite{Kadison1952}; see also \cite{Chruscinski2019}. Let
$\Phi:\mathcal B(\mathcal H)\to\mathcal B(\mathcal H)$ be a positive,
unital, complex-linear map. First, positivity implies Hermiticity
preservation. Indeed, every Hermitian operator $A$ can be written as
$A=A_+-A_-$ with $A_\pm\ge0$, and hence
$\Phi(A)=\Phi(A_+)-\Phi(A_-)$ is Hermitian. Writing an arbitrary
operator as $A=B+iC$, with $B$ and $C$ Hermitian, and using complex
linearity, we obtain
\[
\Phi(A^\dagger)
=
\Phi(B-iC)
=
\Phi(B)-i\Phi(C)
=
\Phi(A)^\dagger .
\]

Now let $X=X^\dagger$. Consider the commutative $C^*$-algebra
$\mathcal A$ generated by $X$ and $\id$. The restriction of $\Phi$ to
$\mathcal A$ is again positive and unital. Since every positive map
from a commutative $C^*$-algebra into $\mathcal B(\mathcal H)$ is
completely positive, the restriction of $\Phi$ to $\mathcal A$ is
completely positive. Therefore, the Schwarz inequality holds on
$\mathcal A$:
\[
\Phi(C^\dagger C)\geq \Phi(C^\dagger)\Phi(C),
\qquad C\in\mathcal A .
\]
Taking $C=X$ gives
\[
\Phi(X^2)\geq \Phi(X)^2 .
\]
This is Kadison's inequality in the self-adjoint form used in the main
text. For arbitrary, not necessarily normal, operators the Schwarz
inequality is guaranteed by unital $2$-positivity, and hence in
particular by unital complete positivity.


\begin{thebibliography}{99}
		
		
		
		
		\bibitem{BreuerPetruccione2002}
		H.-P.~Breuer and F.~Petruccione,
		The theory of open quantum systems,
		OUP Oxford, 2002.
		
		
		
		
		\bibitem{NielsenChuang2010}
		M.~A.~Nielsen and I.~L.~Chuang,
		Quantum computation and quantum information: 10th anniversary edition,
		Cambridge University Press, Cambridge, 2010.
		
		\bibitem{RHP}
		A.~Rivas, S.~F.~Huelga, and M.~B.~Plenio,
		Entanglement and non-Markovianity of quantum evolutions,
		Phys.\ Rev.\ Lett.\ \textbf{105}, 050403 (2010).
		
		\bibitem{BreuerReview}
		H.-P.~Breuer, E.-M.~Laine, J.~Piilo, and B.~Vacchini,
		Colloquium: Non-Markovian dynamics in open quantum systems,
		Rev.\ Mod.\ Phys.\ \textbf{88}, 021002 (2016).
		
		\bibitem{RivasReview}
		A.~Rivas, S.~F.~Huelga, and M.~B.~Plenio,
		Quantum non-Markovianity: Characterization, quantification and
		detection,
		Rep.\ Prog.\ Phys.\ \textbf{77}, 094001 (2014).
		
		\bibitem{WangSchirmer2009}
		X.~Wang and S.~G.~Schirmer,
		Contractivity of the Hilbert--Schmidt distance under open-system dynamics,
		Phys.\ Rev.\ A \textbf{79}, 052326 (2009).
		
		
		
		\bibitem{RanganiJahromi2020}
		H.~Rangani~Jahromi, K.~Mahdavipour, M.~Khazaei~Shadfar, and R.~Lo~Franco,
		Witnessing non-Markovian effects of quantum processes through Hilbert-Schmidt speed,
		Phys.\ Rev.\ A \textbf{102}, 022221 (2020).
		
		\bibitem{BreuerLainePiilo2009}
		H.-P.~Breuer, E.-M.~Laine, and J.~Piilo,
		Measure for the degree of non-Markovian behavior of quantum processes in open systems,
		Phys.\ Rev.\ Lett.\ \textbf{103}, 210401 (2009).
		
		\bibitem{Mahdavipour2022}
		K.~Mahdavipour, M.~Khazaei~Shadfar, H.~Rangani~Jahromi, R.~Morandotti, and R.~Lo~Franco,
		Memory Effects in High-Dimensional Systems Faithfully Identified by Hilbert--Schmidt Speed-Based Witness,
		Entropy \textbf{24}, 395 (2022).
		
		\bibitem{JahromiLoFranco2022}
		H.~Rangani~Jahromi and R.~Lo~Franco,
		Searching for exceptional points and inspecting non-contractivity of
		trace distance in (anti-)$\mathcal{PT}$-symmetric systems,
		Quantum Inf.\ Process.\ \textbf{21}, 155 (2022).
		
		
		
		
			\bibitem{Kadison1952} R.~V.~Kadison, Ann.\ Math.\ \textbf{56}, 494 (1952).
		
		
		
		
		
		\bibitem{GKSL2}
		V.~Gorini, A.~Kossakowski, and E.~C.~G.~Sudarshan,
		Completely positive dynamical semigroups of $N$-level systems,
		J.\ Math.\ Phys.\ \textbf{17}, 821--825 (1976).
		
		\bibitem{GKSL1}
		G.~Lindblad,
		On the generators of quantum dynamical semigroups,
		Commun.\ Math.\ Phys.\ \textbf{48}, 119--130 (1976).
		
		
		
		
		
		
		
		\bibitem{Chruscinski2019} D.~Chru\'sci\'nski and F.~Mukhamedov, Phys.\ Rev.\ A \textbf{100}, 052120 (2019).
		
	\end{thebibliography}
\end{document}